\documentclass[sigconf]{acmart}

\usepackage{multirow}
\usepackage{tabularx}
\usepackage{subfigure}
\usepackage{cleveref}
\usepackage{hyperref}
\usepackage{graphicx}
\usepackage{balance}
\usepackage{natbib}
\usepackage{subcaption}

\AtBeginDocument{%
  }

\copyrightyear{2026}
\acmYear{2026}
\setcopyright{cc}
\setcctype{by}
\acmConference[RecSys '26]{20th ACM Conference on Recommender Systems}{September 27-October 02, 2026}{Minneapolis, MN, USA}
\acmBooktitle{20th ACM Conference on Recommender Systems (RecSys '26), September 27-October 02, 2026, Minneapolis, MN, USA}
\acmDOI{10.1145/3773078.3831832}
\acmISBN{979-8-4007-2284-4/2026/09}

\begin{document}

\title{Give the Long-tail More SPACE: Promoting Provider Fairness in Next POI Recommendation}

\author{Anran Zhang}
\email{arzhang@seu.edu.cn}
\orcid{0009-0005-5350-0281}
\affiliation{%
  \institution{School of Computer Science and Engineering,\\ Southeast University}
  \city{Nanjing}
  \state{Jiangsu}
  \country{China}
}

\author{Jiaqi Jiang}
\email{jiaqijiang@seu.edu.cn}
\orcid{0000-0002-7253-9123}
\affiliation{%
  \institution{School of Computer Science and Engineering,\\ Southeast University}
  \city{Nanjing}
  \state{Jiangsu}
  \country{China}
}

\author{Jiahui Jin}
\authornote{Corresponding Authors.}
\orcid{0000-0001-9570-1456}
\email{jjin@seu.edu.cn}
%\orcid{}
\affiliation{%
  \institution{School of Computer Science and Engineering,\\ Southeast University}
  \city{Nanjing}
  \state{Jiangsu}
  \country{China}
}

\author{Yuhan Zhao}
\authornotemark[1]
\orcid{0000-0002-1427-4139}
\email{csyhzhao@comp.hkbu.edu.hk}
\affiliation{%
  \institution{Department of Computer Science, \\Hong Kong Baptist University}
  \city{Hong Kong}
  \country{China}}

\renewcommand{\shortauthors}{Zhang et al.}

\begin{abstract}
Next point-of-interest (POI) recommendation predicts users’ future destinations from historical mobility sequences and has become a key component of location-based services. However, mainstream models often concentrate exposure on a small set of popular POIs, leaving long-tail merchants systematically under-exposed. While provider fairness has recently attracted increasing attention, directly applying existing provider-fairness techniques to POI recommendation is problematic: (i) users face \emph{execution constraints}; and (ii) POIs face \emph{resource supply constraints}. 
% These coupled constraints render provider fairness in POI recommendation a fundamentally different—and more challenging—problem than in purely digital settings. 
To address this, we propose \textsc{SPACE} (\textbf{S}upply- and \textbf{P}hysics-\textbf{A}ware \textbf{C}onditional \textbf{E}mbedding generation), a model-agnostic framework that improves long-tail POI exposure via \emph{virtual user generation} under explicit feasibility and supply control. \textsc{SPACE} consists of three stages: (1) community inference to capture heterogeneous user execution constraints; (2) unbalanced optimal-transport allocation to decide how many virtual users each tail POI should receive from which communities under POI-specific supply budgets; and (3) constraint-guided latent diffusion to generate POI-conditional, community-consistent virtual user embeddings. The generated user--POI pairs can be seamlessly used to train existing recommenders without modifying their architectures. Extensive experiments on three real-world datasets demonstrate that \textsc{SPACE} substantially improves provider fairness while maintaining—and often improving—recommendation accuracy across multiple backbone models. Our code is publicly available at \url{https://github.com/Anniran1/SPACE-main}.
\end{abstract}

\begin{CCSXML}
<ccs2012>
   <concept>
       <concept_id>10002951</concept_id>
       <concept_desc>Information systems</concept_desc>
       <concept_significance>500</concept_significance>
       </concept>
   <concept>
%        <concept_id>10002951.10003317.10003347.10003350</concept_id>
       <concept_desc>Information systems~Recommender systems</concept_desc>
       <concept_significance>500</concept_significance>
       </concept>
 </ccs2012>
\end{CCSXML}
\ccsdesc[500]{Information systems}
\ccsdesc[500]{Information systems~Recommender systems}

\keywords{Next Point-of-interest Recommendation, Fairness}


\maketitle

\section{Introduction}
Next point-of-interest (POI) recommendation aims to predict users' future destinations from their behavioral sequences. While modern recommenders \cite{huang2025poi,sun2024going} enhance user experience by delivering accurate suggestions, they can simultaneously impose an unintended cost on long-tail merchants whose exposure falls far below the average \cite{yin2012challenging}, resulting in systematically inequitable visibility. In response, \emph{provider fairness} has emerged as a central research topic, seeking to democratize exposure opportunities for underrepresented providers \cite{guo2025enhancing, singh2018fairness}. 
\begin{figure}[t]
    \centering
    \includegraphics[width=\columnwidth]{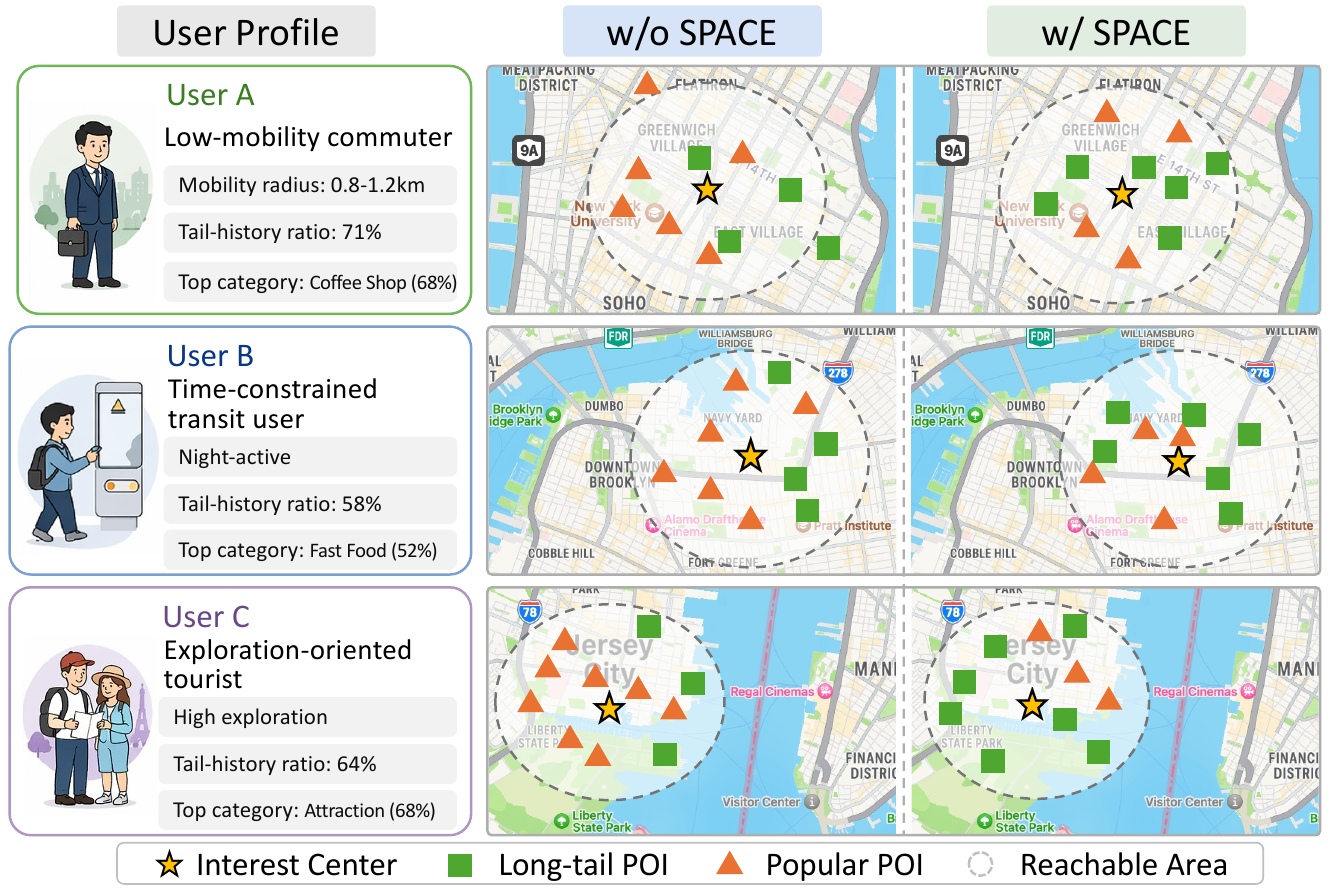}
    \caption{Recommendation case studies on three types of constrained users. The circles on the maps define the reachable radius for constrained users. The original model (w/o SPACE) tends to concentrate exposure on popular POIs, even when such recommendations are not fully aligned with users’ constraints. In contrast, after incorporating SPACE (w/SPACE), the recommendation list contains more long-tail POIs that are both consistent with user preferences and physically reachable within their mobility budgets.}
    \label{fig:case study}
\end{figure}

Beyond algorithmic parity, provider fairness in POI settings is intertwined with \emph{spatial justice} and \emph{urban sustainability}. As depicted in \figurename~\ref{fig:case study}, the conventional next POI recommendation model tends to recommend popular POIs with high historical exposure for users. Such feedback loops amplify "rich-get-richer" effects, funneling crowds into hotspots while accelerating regional polarization. These bias also confines users to homogenized, overcrowded enclaves, creating informational silos that hinder the discovery of local culture. Therefore, ensuring provider fairness in POI recommendation is not merely an optimization preference but a sociotechnical imperative. However, directly transplanting existing provider-fairness techniques into the POI domain is ill-suited. Unlike purely digital recommendation settings \cite{chen2025leave, zhao2026double}, POI recommendation is governed by two distinctive constraints:
\begin{itemize}
\item \textbf{User execution constraints.} Unlike online actions (e.g., clicks), physical visits require tangible resources—time, energy, and money—which are unevenly distributed across users. Consequently, the burden of a “fair but unsuitable” recommendation is inherently asymmetric: a car owner may accommodate a detour with minimal friction, whereas a tourist relying on public transportation can incur substantial time loss and utility degradation if redirected to a distant location merely to satisfy a fairness metric.
\item \textbf{Resource supply constraints.} Physical venues operate under hard capacity and availability limits. For instance, niche museums may have limited ticket quotas due to space constraints. Naively boosting the exposure of tail POIs to achieve statistical fairness can thus lead to over-recommendation, creating congestion that degrades visitor experience and overwhelms providers’ operational capacity.
\end{itemize}

These constraints make POI provider fairness uniquely complex, meaning conventional fairness methods can severely compromise accuracy. For instance, directly promoting certain providers through re-ranking may ignore physical feasibility or supply limits, leading to recommendations that users cannot visit or providers cannot accommodate. Such mismatches ultimately feed back into degraded user satisfaction and impaired system accuracy—outcomes we explicitly seek to avoid. 

To address this challenge, we propose \textsc{SPACE} (\textbf{S}upply- and \textbf{P}hysics-\textbf{A}ware \textbf{C}onditional \textbf{E}mbedding generation), a model-agnostic framework that improves long-tail POI exposure by \emph{generating physically plausible virtual users} under explicit supply- and execution-aware constraints. Instead of manipulating rankings in a way that may conflict with real-world feasibility, \textsc{SPACE} augments training signals for tail POIs through constraint-respecting interaction synthesis: it (i) infers user communities that reflect heterogeneous mobility/execution constraints, (ii) allocates augmentation quotas to tail POIs under explicit supply budgets via unbalanced optimal transport, and (iii) generates POI-conditional virtual user embeddings through constraint-guided latent diffusion, so that downstream recommenders can be trained with improved fairness while preserving (and often improving) accuracy, as shown in \figurename~\ref{fig:case study}.

Our main contributions are summarized as follows:
\begin{itemize}
\item We formalize \emph{provider fairness} in POI recommendation as a \emph{constraint-coupled} problem. We show that, unlike conventional recommendation settings, both users and POIs in POI recommendation are subject to intrinsic physical and supply limitations, which necessitate rethinking provider fairness for this domain and lay a foundation for future work.
\item We propose \textsc{SPACE}, a supply- and physics-aware conditional embedding generation framework that democratizes tail-POI exposure via virtual user generation. \textsc{SPACE} integrates community inference, unbalanced-OT-based supply-constrained quota allocation, and constraint-guided latent diffusion to address domain-specific physical and supply constraints.
\item Extensive experiments on three real-world datasets demonstrate that \textsc{SPACE} substantially improves provider fairness for various backbone recommenders \emph{without sacrificing accuracy}—and in most cases even \emph{improves} it. These impressive results highlight strong practical value for real POI platforms.
\end{itemize}
\section{Preliminaries}
\label{sec:problem_formulation}

In this section, we formally define the provider fairness in next-POI recommendation.
We then theoretically demonstrate why fairness paradigms based on unconstrained continuous optimization can become infeasible and even harmful.

\subsection{Problem Definition}

\textbf{Next-POI Recommendation.}
Let $\mathcal{U}=\{u_1,\ldots,u_{|\mathcal U|}\}$ be the set of users and $\mathcal{V}=\{v_1,\ldots,v_{|\mathcal V|}\}$ be the set of POIs.
Each POI $v\in\mathcal V$ is associated with geographic coordinates and semantic attributes. A check-in is a tuple $q=(u,v,t)$ indicating that user $u$ visits POI $v$ at timestamp $t$. For a user $u$, their trajectory is $\mathcal{T}_u=\{q_u^1,\ldots,q_u^{n_u}\}$.
Given $\mathcal{T}_u$, the next POI recommendation aims to predict a list of possible POIs that the user is inclined to visit next. We denote the predicted matching utility between $u$ and $v$ as $\hat{y}_{u,v}\in\mathbb{R}^{+}$.

\noindent\textbf{Provider Fairness in POI Recommendation.}
Provider fairness aims to mitigate exposure bias against marginalized long-tail POIs. Let $x_{u,v}\in\{0,1\}$ indicate whether POI $v$ is included in user $u$'s top-$K$ recommendation list.
Thus, each recommendation request must satisfy
$\sum_{v\in\mathcal V} x_{u,v}\le K$.
Let $\mathcal{V}_{\text{tail}}\subset\mathcal V$ denote the set of long-tail POIs, following the definition in existing research\cite{yin2012challenging}. We use a simple exposure proxy
$E(v)=\sum_{u\in\mathcal U} x_{u,v}$,
and define the average exposure of non-tail POIs as
$\bar{E}=\frac{1}{|\mathcal V\setminus\mathcal V_{\text{tail}}|}\sum_{v\in\mathcal V\setminus\mathcal V_{\text{tail}}}E(v)$.
A common exposure-based provider-fairness requirement enforces that each tail POI should receive at least a fraction of this average exposure:
\begin{equation}
    E(v)=\sum_{u\in\mathcal U} x_{u,v} \ge \tau\cdot \bar{E},\quad \forall v\in\mathcal V_{\text{tail}},
\end{equation}
where $\tau\in(0,1]$ controls the strictness of the fairness requirement.

\subsection{Physical Constraints and Realized Utility}

Unlike purely digital platforms, POI recommendation is constrained by physical feasibility on both the user side and the POI side.

\begin{itemize}
    \item \textbf{User execution constraint.}
    Let $c_{u,v}$ denote the execution cost for user $u$ to visit POI $v$ (e.g., travel time/distance under a transportation mode), and let $B_u$ be the user's execution budget.
    A recommendation is executable only if $c_{u,v}\le B_u$.
    \item \textbf{Resource supply constraint.}
    Let $S_v$ denote the service/supply capacity of POI $v$ over the considered horizon (e.g., maximum acceptable visits or ticket capacity).
    A physically meaningful recommendation allocation should respect
    $\sum_{u\in\mathcal U} x_{u,v}\le S_v$.
\end{itemize}

These constraints distinguish \emph{digital expected utility} from \emph{physical realized utility}:
even if a model assigns high $\hat{y}_{u,v}$, recommending an unexecutable POI produces ``fake exposure'' that cannot translate into real visits.

\begin{definition}[Realized System Utility]
Given an allocation matrix $\mathbf X=[x_{u,v}]$, the realized system utility is
\begin{equation}
U_{\mathrm{real}}(\mathbf X)
=
\sum_{u\in\mathcal U}\sum_{v\in\mathcal V}
\hat{y}_{u,v}\, x_{u,v}\, \mathbb{I}(c_{u,v}\le B_u),
\end{equation}
where $\mathbb{I}(\cdot)$ is the indicator function.
The goal is to maximize $U_{\mathrm{real}}(\mathbf X)$ subject to constraints and fairness requirements.%top-$K$ constraints, supply constraints, and fairness requirements.
\end{definition}

\subsection{Theoretical Analysis of Existing Paradigms}

While many fair recommenders rely on soft regularization and end-to-end optimization, we show that POI recommendation is a combinatorial, NP-hard problem under physical constraints. Consequently, ignoring execution feasibility can lead to unbounded utility degradation.

\begin{theorem}[NP-Hardness of Physically Feasible Fair Allocation]
\label{thm:intractability}
With heterogeneous execution budgets $\{B_u\}$ and POI capacities $\{S_v\}$, the problem of finding an allocation that maximizes realized utility subject to physical feasibility is NP-hard; adding exposure-based provider fairness constraints does not simplify the problem.
Consequently, continuous soft-penalty optimization does not provide feasibility guarantees for the resulting discrete top-$K$ allocations.
\end{theorem}

\begin{proof}
We prove NP-hardness by reduction from a capacitated assignment problem.
Consider a simplified setting where each user must be assigned to exactly one POI (i.e., set $K=1$), and ignore the fairness constraints (removing constraints cannot make the problem harder).
Let $w_{u,v}=\hat{y}_{u,v}\cdot \mathbb{I}(c_{u,v}\le B_u)$ be the realized profit of assigning $u$ to $v$.
Then the decision problem becomes:
\begin{align}
\max_{\mathbf X}\ &\sum_{u\in\mathcal U}\sum_{v\in\mathcal V} w_{u,v}x_{u,v} \\
\quad
\text{s.t.}\quad
&\sum_{v\in\mathcal V}x_{u,v}=1,\ \forall u;\qquad \\
&\sum_{u\in\mathcal U}x_{u,v}\le S_v,\ \forall v;\qquad
x_{u,v}\in\{0,1\}.
\end{align}
This is a standard capacitated assignment / generalized assignment formulation, which is NP-hard in general.
Therefore, our original problem with top-$K$ recommendation ($K\ge 1$) and additional fairness constraints is also NP-hard.
\end{proof}

\begin{theorem}[Unbounded Price of Fairness Under Execution Infeasibility]
\label{thm:pof}
Enforcing exposure-based provider fairness without explicitly accounting for user execution feasibility $\mathbb{I}(c_{u,v}\le B_u)$ can yield an unbounded Price of Fairness (PoF) in physical space, leading to a potentially catastrophic drop in realized utility.
\end{theorem}

\begin{proof}
Define $\mathrm{PoF}=U_{\mathrm{OPT}}/U_{\mathrm{FAIR}}$, where $U_{\mathrm{OPT}}$ is the maximum realized utility without fairness constraints and $U_{\mathrm{FAIR}}$ is the realized utility under a fairness-enforced policy.

Consider an LBSN with $M$ users, one popular POI $v_{\mathrm{pop}}$, and one tail POI $v_{\mathrm{tail}}\in\mathcal V_{\mathrm{tail}}$.
Assume all users can execute visits to $v_{\mathrm{pop}}$ but none can execute visits to $v_{\mathrm{tail}}$:
$c_{u,v_{\mathrm{pop}}}\le B_u$ and $c_{u,v_{\mathrm{tail}}}>B_u$ for all $u$.
Let $\hat y_{u,v_{\mathrm{pop}}}=1$ and $\hat y_{u,v_{\mathrm{tail}}}=1-\epsilon$ with $\epsilon\to 0^{+}$. Without fairness constraints, recommending $v_{\mathrm{pop}}$ to all users yields:
\begin{equation}
U_{\mathrm{OPT}}= \sum_{u=1}^M 1\cdot \mathbb{I}(c_{u,v_{\mathrm{pop}}}\le B_u)=M.
\end{equation}

Now impose an exposure-based fairness constraint that forces a fraction $\tau\in(0,1)$ of recommendations to be allocated to $v_{\mathrm{tail}}$.
In physical space, those recommendations are infeasible and contribute zero realized utility:
\begin{equation}
U_{\mathrm{FAIR}}^{\mathrm{Physical}}
=
(1-\tau)M\cdot 1 + \tau M\cdot (1-\epsilon)\cdot 0
=
(1-\tau)M.
\end{equation}
Hence,
\begin{equation}
\mathrm{PoF}^{\mathrm{Physical}}
=
\frac{U_{\mathrm{OPT}}}{U_{\mathrm{FAIR}}^{\mathrm{Physical}}}
=
\frac{M}{(1-\tau)M}
=
\frac{1}{1-\tau}
\;\xrightarrow[\tau\to 1]{}\;+\infty.
\end{equation}

Therefore, blindly boosting exposure for tail POIs without feasibility awareness can force ``unexecutable'' recommendations, producing fake exposures that do not convert into visits, and can arbitrarily harm realized utility.
\end{proof}
\section{Methodology}
To address the complex \emph{provider fairness} problem in POI recommendation, we propose \textsc{SPACE} (\textbf{S}upply- and \textbf{P}hysics-\textbf{A}ware \textbf{C}onditional \textbf{E}mbedding generation), a framework that democratizes long-tail POI exposure via \textit{virtual user generation}.
% \textsc{SPACE} increases interactions for tail POIs by generating virtual users that satisfy two key constraints in POI scenarios: \emph{resource supply capacity} on the POI side and \emph{execution feasibility} on the user side.
\textsc{SPACE} consists of three stages:
(1) \emph{Community prototype inference}, which partitions users into different communities according to their inherent constraints;
(2) \emph{Supply-constrained allocation}, which adaptively decides how many virtual users each tail POI should receive and from which communities;
and (3) \emph{Constraint-guided generation}, which produces the final virtual user embeddings conditioned on the target POI and the allocated community.
% Existing recommender models can directly train on these generated user-POI embedding pairs, enabling seamless and model-agnostic integration.
The overall workflow is depicted in \figurename~\ref{fig:SPACE}.
\begin{figure*}[t]
    \centering
    \includegraphics[trim=0cm 4.5cm 0cm 4.5cm, clip, width=0.85\textwidth]{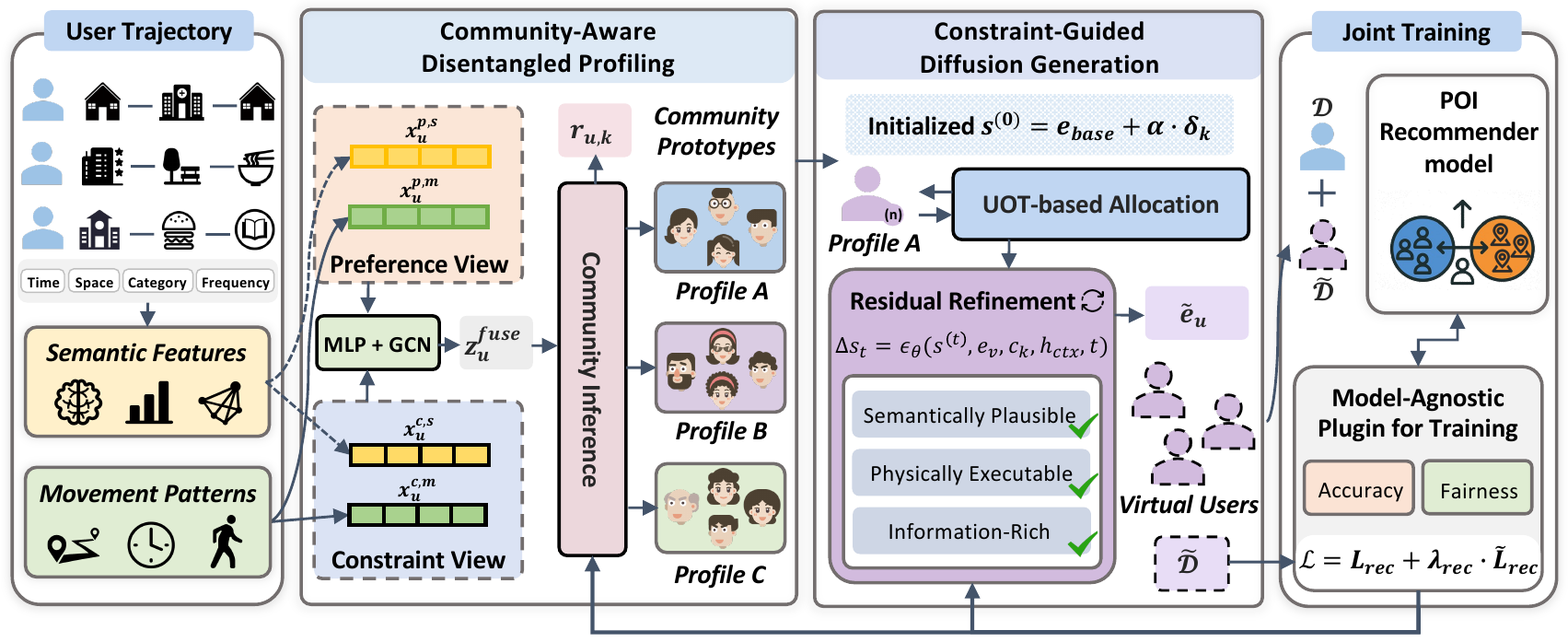}
    \caption{The illustration of our proposed \textsc{SPACE}.}
    \label{fig:SPACE}
\end{figure*}

\subsection{Latent Constraint Discovery through Community Inference}
Inspired by prior community inference studies \cite{zhang2025hacd,gao2025social,chen2025communitydf}, we divide users into different communities according to their inherent constraints. Intuitively, we expect this module to learn communities with the following properties: (1) users assigned to the same community are relatively consistent in constraints; (2) we allow correlations between preference and constraints, but they are not fully entangled in a single community (users with similar preferences may still be separated due to different constraints).
With such communities, our later modules can leverage them to guide the generation of physically plausible virtual users.

\subsubsection{Community Assignment}
To achieve this goal, we propose a novel community inference function.
First, we represent each user from two complementary views.
Given user $u$'s raw behavioral features, we summarize them into:
(i) the \textbf{preference view} $\mathbf{x}_u^p$, which captures the user's preference profile using a normalized POI-category histogram, the total number of visits, the number of unique visited POIs, and an exploration diversity ratio; and
(ii) the \textbf{constraint view} $\mathbf{x}_u^c$, which captures the user's mobility pattern using a normalized hourly activity histogram, a day-of-week histogram, geographic statistics of visited locations, and the average trajectory length. The former mainly reflects historical POI choices (preference), while the latter mainly reflects mobility regularities (constraints). Then, we map the two views into one \textbf{constraint latent} embedding $\mathbf{z}_u^c$, which represents the user's constraints: %characteristics:
\begin{equation}
\mathbf{z}_u^c=\mathrm{MLP}_{c}\!\left([\mathbf{x}_u^c;\mathbf{x}_u^p]\right).
\end{equation}
Notably, we use both $\mathbf{x}_u^p$ and $\mathbf{x}_u^c$ to build $\mathbf{z}_u^c$ because preference and constraints may influence each other. For example, stronger mobility budgets often lead to more exploratory preferences, and ``nightlife'' preferences imply late-hour activity patterns. This design allows us to capture the intricate interplay between constraint and preference signals.

Moreover, collaborative signals are critical in recommendations. Instead of modeling each user independently, we utilize a GNN \cite{kipf2016semi} to provide additional signals:
\begin{equation}
    \mathbf{g}_u^c=\mathrm{GCN}(G,\mathbf{z}_u^{c}),
\end{equation}
where $G$ is a user graph constructed from user interactions. 

Based on this, we fuse multiple sources to obtain the representation for community partition:
\begin{equation}
    \mathbf{z}_u^{fuse}=\mathrm{MLP}_{fuse}\!\left([\mathbf{z}_u^{c};\mathbf{g}_u^{c}]\right).
\end{equation}

Next, we adopt a commonly used soft clustering approach \cite{shchur2019overlapping} to assign users into $K$ communities.
We initialize a set of learnable community prototypes $\mathbf{C} = \{\mathbf{c}_1, \ldots, \mathbf{c}_K\}$.
The posterior soft community assignment is computed as:
\begin{equation}
r_{u,k}=
\frac{\exp\left(\mathrm{sim}(\mathbf{z}_u^{fuse},\mathbf{c}_k)/\tau\right)}
{\sum_{j=1}^{K}\exp\left(\mathrm{sim}(\mathbf{z}_u^{fuse},\mathbf{c}_j)/\tau\right)},
\end{equation}
where $\mathrm{sim}(\cdot,\cdot)$ denotes cosine similarity and $\tau$ is a temperature hyper-parameter.
Users can then be assigned to communities according to these scores.

\subsubsection{Inference Loss}
Although the above scheme is intuitive, there is a critical issue: the recommender's objective is designed for recommendation accuracy rather than directly supervising the community inference task.
Therefore, we design an auxiliary loss to support community inference from the marginal community usage:
\begin{equation}
\tilde{\mathbf r}_k=\frac{1}{|\mathcal U|}\sum_{u\in\mathcal U} r_{u,k}.
\end{equation}
%a loss function 
Then, we propose to learn compact and non-collapsed communities:
\begin{equation}
\arg \min_{\Theta_{\text{CI}}}\ \mathcal{L}_{\text{infer}}
=
-\frac{1}{|\mathcal{U}|}\sum_{u\in \mathcal{U}} \sum_{k=1}^K
r_{u,k}\cdot \frac{\mathrm{sim}(\mathbf{z}_u^{fuse},\mathbf{c}_k)}{\tau}
+
D_{\mathrm{KL}}\!\left(\tilde{\mathbf r}\,\middle\|\,U(1/K)\right),
\end{equation}
where $U(1/K)$ denotes the uniform distribution over $K$ communities. The first term maximizes the expected similarity between each user representation and its assigned prototype(s), thereby forming tight and coherent communities in the latent space.
The second term prevents trivial solutions where all users collapse into a single community by encouraging the \emph{marginal} community distribution $\tilde{\mathbf r}$ to be close to uniform.
Importantly, this regularizer does not enforce hard per-user assignments; it only discourages degenerate global partitions, which is sufficient for maintaining community coverage in downstream quota allocation and executable virtual user generation.

Since this loss is only meant to regularize the community inference process, we restrict parameter updates to community-inference parameters $\Theta_{\text{CI}}$ (i.e., $\mathrm{MLP}_c$, $\mathrm{GCN}$, $\mathrm{MLP}_{fuse}$, and prototypes $\mathbf C$).

\subsection{Supply-Constrained Tail POI Allocation via Unbalanced Optimal Transport}
After obtaining communities, a natural question is: can we immediately start generating virtual users to increase interactions for tail POIs?
The answer is no. If we blindly augment all tail POIs (e.g., generating a fixed number $M$ of virtual users for each tail POI), two serious biases may occur: (1) it violates the \emph{Resource Supply Constraint} because real POIs have physical/service capacity limits; generating too many ``virtual visits'' creates unrealistic demand and may mislead the model to over-recommend niche POIs; (2) it causes distribution distortion: injecting too many positive interactions may turn a tail POI into a ``fake head'' item, biasing the ranking model and harming overall accuracy and user experience.

Thus, we cannot pursue fairness while ignoring practical constraints.
We need a carefully designed allocation strategy that determines how many virtual users each tail POI should receive and from which communities.
We set a supply budget vector $\mathbf{d}$, where each tail POI $v$ should receive around $d_v$ (e.g., historical maximum daily visits) virtual users, drawn from multiple relevant communities.

To achieve this, we design an Optimal Transport (OT) \cite{wu2021tfrom,cuturi2013sinkhorn} based allocation scheme.
Let $\mathbf{P} \in \mathbb{R}_{+}^{|\mathcal{V}_{tail}| \times K}$ be the transport plan, where $P_{v,k}$ indicates the number of virtual embeddings to generate for POI $v$ from community $k$.
To allocate only from relevant communities, we define a cost matrix $\mathbf{Cost}$ that measures the distance between POI embeddings $\mathbf{e}_v$ and community representations $\boldsymbol{\mu}_k$.
The OT objective is:
\begin{equation}
\min_{\mathbf{P} \geq 0} \ \langle \mathbf{P}, \mathbf{Cost} \rangle
+ \underbrace{\mathrm{KL}\left(\mathbf{P} \mathbf{1} \, \middle\| \, \mathbf{d}\right)}_{\text{Supply Constraint}}.
\end{equation}
The first term encourages allocating a POI's quota to more relevant communities: if a community has a lower cost (higher similarity) to the POI, then $P_{v,k}$ tends to be larger.
The second term encourages the total allocation for each POI to be close to $d_v$; deviations are allowed but penalized, and larger deviations incur larger penalties.
With this allocation, each tail POI receives a reasonable mixture of communities and is pulled back to its budget $d_v$, avoiding becoming a ``fake head'' POI.
For example, a niche art museum will not be forced to match an irrelevant ``scenery-preferring tourist'' community because that would incur a high cost, while its total quota will be kept around $d_v$ rather than exploding unrealistically.

\subsection{Constraint-Guided Latent Diffusion for Virtual User Embeddings}
\label{sec:diffusion}
With the community and quota allocation determined, we generate \textbf{Virtual User Embeddings} $\tilde{\mathbf{e}}_u$ for each tail POI $v$.
Direct augmentation strategies often fail to satisfy the coupled constraints: generated users must be (i) \emph{semantically plausible} (aligned with POI $v$ and community $k$), (ii) \emph{information-rich} (avoiding mode collapse or meaningless noise), and (iii) \emph{physically executable} (respecting the community's execution characteristics).
To address these challenges, we propose a \textbf{POI-conditional latent generator} inspired by diffusion models \cite{ho2020denoising,liu2026diffgrm}.
Instead of sampling from isotropic noise, we perform \emph{iterative residual denoising} starting from a reliable initialization, decomposing the complex generation process into multiple small, constraint-aware refinements.

\subsubsection{Prototype-Aware Initialization}
A robust generation process requires a starting point that already lies in a plausible region of the target community.
For a given pair $(v,k)$, we construct a community-aware base state by fusing the POI representation $\mathbf{h}_v$ with the community prototype $\mathbf{c}_k$:
\begin{equation}
\mathbf{e}_{\text{base}} = \mathrm{MLP}_{\text{base}}\!\left([\mathbf{h}_v; \mathbf{c}_k]\right).
\end{equation}
Here, $\mathbf{c}_k$ serves as the sole community prototype, encoding the execution characteristics of community $k$ learned in the community inference stage.

To inject diversity grounded in real user variations, we sample a real user $u$ from community $k$ and compute its deviation from the community center in the constraint latent space:
\begin{equation}
\delta_k = \mathbf{z}_u^c - \bar{\mathbf{z}}_k^c, \quad 
\bar{\mathbf{z}}_k^c = \mathbb{E}[\mathbf{z}_j^c \mid j \in k], \quad
u \sim \text{Unif}(\{j: j \in k\}).
\end{equation}
The initial latent state is defined as:
\begin{equation}
\mathbf{s}^{(0)} = \mathbf{e}_{\text{base}} + \alpha \cdot \delta_k,
\end{equation}
where $\alpha$ controls exploration strength. This initialization anchors generation around $(v,k)$ while exploring along \emph{real} intra-community directions, avoiding the need to learn feasibility from scratch.

\subsubsection{Iterative Residual Refinement}
We introduce a conditional denoising network $\epsilon_\theta$ to iteratively refine the latent state $\mathbf{s}^{(t)}$ for $t=0,\dots,T-1$.
At each step, the network predicts a residual update conditioned on the current state, POI context, community prototype, and batch-level statistics:
\begin{equation}
\Delta \mathbf{s}_t = \epsilon_\theta\!\left(\mathbf{s}^{(t)},\, \mathbf{e}_v,\, \mathbf{c}_k,\, \mathbf{h}_{\text{ctx}},\, t\right),
\end{equation}
where $\mathbf{h}_{\text{ctx}}$ summarizes the distribution of real user embeddings in the current batch (e.g., mean pooling), preventing the generated embeddings from drifting away from the global latent geometry.
We update the state with a constant step size $\eta_t = 1/T$:
\begin{equation}
\mathbf{s}^{(t+1)} = \mathbf{s}^{(t)} + \eta_t \Delta \mathbf{s}_t.
\end{equation}
After $T$ refinement steps, the final virtual user embedding is obtained as:
\begin{equation}
\tilde{\mathbf{e}}_u = \mathbf{s}^{(T)}.
\end{equation}

In addition, we require the generated virtual users to remain consistent with the constraint characteristics of their source community.
We impose the following execution regularization:
\begin{equation}
\arg \min_{\Theta_{\text{Gen}}} \mathcal L_{\text{exec}}
=
\mathrm{ReLU}\!\left(\|\tilde{\mathbf e}_u -\mathbf{e}_u\|_2^2 - R_k^2\right),
\end{equation}
Here, $R_k$ is computed by a robust statistic:
\begin{equation}
R_k^2 \;=\; \operatorname{Quantile}_{q}\Big(\big\{\|\mathbf{e}_u-\bar{\mathbf{e}}_k\|_2^2:\; u\in k\big\}\Big),
\end{equation}
where $\bar{\mathbf e}_k$ denotes the community center.
This loss is essentially a ``community constraint'': if the generated user is close to community $k$, the generation has more freedom; otherwise, it receives a strong penalty, pushing the generation back to the community manifold and reducing implausible samples that violate community characteristics. Since this loss is only meant for the generation process, we restrict parameter updates to generation parameters $\Theta_{\text{Gen}}$.

\subsection{Training}
After generation, we pair the virtual users with their corresponding tail POIs to form a new dataset $\tilde{\mathcal{D}}$, and train together with the original dataset $\mathcal{D}$:
\begin{equation}
\mathcal{L} =
\sum_{(u, i) \in \mathcal{D}} L_{\text{rec}}(u,i) 
+ \lambda_{\text{rec}}\sum_{(u, i) \in \tilde{\mathcal{D}}} L_{\text{rec}}(u,i),
\end{equation}
where $\lambda_{\text{rec}}$ balances the contribution of the augmented data.
\section{Experiments}
In this section, we conduct comprehensive experiments to evaluate the performance of SPACE. Specifically, we evaluate the effectiveness of integrating SPACE in enhancing long-tail POIs' exposure across various next POI recommendation models (\ref{Sec. 4.2}). We analyze the contribution of individual SPACE components to model performance (\ref{Sec. 4.3}). We also investigate the impact of different parameter settings on SPACE’s efficacy (\ref{Sec. 4.4}) and assess its efficiency for practical applications (\ref{Sec. 4.5}). Additionally, we supplement the results on extra fairness metrics (\ref{Sec. 4.6}) and provide case studies (\ref{Sec. 4.7}).  
\subsection{Experimental Setting}
\subsubsection{Datasets.} We conduct experiments on three publicly POI datasets collected from location-based service platforms: NYC, TKY \cite{yang2014modeling}, and CA \cite{yuan2013time}. All three datasets collect check-ins from real-world cities. The statistics are summarized in \tablename~\ref{tab:data}.
\begin{table}
\setlength{\tabcolsep}{6pt}
    \footnotesize
  \centering
  \caption{Statistics of the POI recommendation datasets.}
  \vspace{-4pt}
  \label{tab:data}
  \begin{tabular}{cccccc}
    \toprule
    Dataset&\#User&\#Poi&\#Cat&\#Check-in&\#Trajectory\\
    \midrule
    NYC & 1,075 & 5,099 & 318 & 104,074 & 14,160\\
    TKY & 2,281 & 7,844 & 291 & 361,430 &44,692\\
    CA & 4,318 & 9,923 &301 &250,780 &32,920\\
  \bottomrule
\end{tabular}
\end{table}
\subsubsection{Baseline.} To validate the effectiveness of our solution, we integrate SPACE with a wide range of representative and state-of-the-art methods. \textbf{FPMC} \cite{rendle2010factorizing} and \textbf{LSTPM} \cite{zhao2020go} utilize Matrix Factorization and LSTM, respectively, to directly model user check-in sequences. \textbf{GETNext} \cite{yang2022getnext} employs a graph-enhanced Transformer on a user-agnostic global trajectory flow map to capture collaborative signals. \textbf{MTNet} \cite{huang2024learning} and \textbf{DiffuRec} \cite{li2023diffurec} employ generative techniques or diffusion frameworks to capture the distribution of potential next POIs. 
\subsubsection{Evaluation metrics.} To evaluate accuracy employing standard metrics for top-K recommendations, we encompass hit rate ($\text{HR@}K$) and normalized discounted cumulative gain ($\text{NDCG@}K$). For fairness measurement, we use cold-warm group fairness ($\text{CGF@}K$), cold POI exposure ($\text{CE@}K$), and long-tail coverage (LTC) to evaluate the inequality in recommendation quality between warm POIs and long-tail POIs. 
\subsubsection{Implementation Details.}  
To ensure reproducibility, we implement all methods and apply pre-processing, dataset split, and evaluation using the POI recommendation standard settings \cite{yang2022getnext}. We construct a popularity ranking by frequency, with the top 20\% of POIs considered popular and the bottom 80\% long-tail. We meticulously tune all hyperparameters on the validation sets and report the best performance for all baseline models.
\subsection{Overall Performance Comparison}\label{Sec. 4.2}
We report the overall performance results in \tablename~\ref{tab:overall}. The key observations include:
\begin{itemize}
    \item Integrating SPACE into diverse backbones leads to substantial improvements in fairness metrics across nearly all datasets. This demonstrates the effectiveness of SPACE in increasing the exposure opportunities of long-tail POIs and promotes fairer recommendations. Such enhancements not only bolster the economic viability of niche merchants but also enrich user satisfaction by offering more diverse options. 
    \item Beyond fairness, SPACE also yields significant performance gains on most datasets and metrics. For instance, GetNext shows a 54.85\% improvement on CA. This illustrates that our synthesized virtual users effectively adhere to the dual constraints, thereby rectifying exposure bias while simultaneously elevating overall recommendation quality.
    \item Notably, the fairness improvement for LSTPM is more modest compared to other baselines. This is primarily because LSTPM inherently mitigates exposure bias by explicitly modeling long- and short-term interests alongside sequential transition patterns.
\end{itemize}
\begin{table*}[t]
\centering
\caption{Experimental results of different methods with (w) or without (w/o) our SPACE plugin (p <= 0.05).}
\vspace{-4pt}
\label{tab:overall}
\scriptsize
\setlength{\tabcolsep}{4pt}
\renewcommand{\arraystretch}{0.8}
\resizebox{\textwidth}{!}{%
\begin{tabular}{ll*{10}{c}}
\toprule
\multirow{2}{*}{Dataset} & \multirow{2}{*}{Metric} &
\multicolumn{2}{c}{FPMC} &
\multicolumn{2}{c}{LSTPM} &
\multicolumn{2}{c}{GETNext} &
\multicolumn{2}{c}{MTNet} &
\multicolumn{2}{c}{DiffuRec}\\
\cmidrule(lr){3-4}\cmidrule(lr){5-6}\cmidrule(lr){7-8}\cmidrule(lr){9-10}\cmidrule(lr){11-12}
& & w/o & w & w/o & w & w/o & w & w/o & w & w/o & w  \\
\midrule

\multirow{10}{*}{NYC}
& HR@1 $\uparrow$ &0.0963  &\textbf{0.1135}    &0.1911  &\textbf{0.1926}  &0.2219  &\textbf{0.2791}  &0.2139  &\textbf{0.2257}  &0.2127  &\textbf{0.2147}  \\
& HR@10 $\uparrow$ &0.4922  & \textbf{0.4983}    &0.5574  &0.5544  &0.5676  &\textbf{0.6638}  &0.5381  &\textbf{0.5537}  &0.5378  &0.5358 \\
& NDCG@10 $\uparrow$ &0.2742  &\textbf{0.2838}    &0.3608  &\textbf{0.3610}  &0.3790  &\textbf{0.4621}  &0.3671  &\textbf{0.3836}  &0.3752  &0.3730\\
\cmidrule(lr){2-12}
& Accuracy Improvement
& \multicolumn{2}{c}{\textbf{7.53\%}}
& \multicolumn{2}{c}{\textbf{0.10\%}}
& \multicolumn{2}{c}{\textbf{21.55\%}}
& \multicolumn{2}{c}{\textbf{4.30\%}}
& \multicolumn{2}{c}{-0.01\%}\\
\cmidrule(lr){2-12}
& CGF@1 $\downarrow$ &0.0521  &\textbf{0.0456}   &0.0824  &\textbf{0.0814}  &0.0372  &\textbf{0.0295}  &0.0667  &\textbf{0.0618}  &0.0200  &\textbf{0.0178}  \\
& CGF@10 $\downarrow$ &0.0448  &\textbf{0.0415}   &0.0738  &\textbf{0.0718}  &0.0265  &\textbf{0.0186}  &0.0621  &\textbf{0.0561}  &0.0206  &\textbf{0.0200}  \\
& CE@1 $\uparrow$ &0.0884  &\textbf{0.0949}   &0.0457  &\textbf{0.0469}  &0.0188  &\textbf{0.0209}  &0.0091  &\textbf{0.0104}  &0.0070  &\textbf{0.0081}  \\
& CE@10 $\uparrow$ &0.4618  &\textbf{0.4725}  &0.2494  &\textbf{0.2535}  &0.0999  &\textbf{0.1100}  &0.0470  &\textbf{0.0543}  &0.0317  &\textbf{0.0357}  \\
& LTC $\uparrow$ &0.5469  &\textbf{0.5858}   &0.5303  &\textbf{0.5339}  &0.7306  &\textbf{0.7686}  &0.3730  &\textbf{0.4234}  &0.6913  &\textbf{0.7284}  \\
\cmidrule(lr){2-12}
& Fairness Improvement
& \multicolumn{2}{c}{\textbf{7.32\%}}
& \multicolumn{2}{c}{\textbf{1.77\%}}
& \multicolumn{2}{c}{\textbf{15.40\%}}
& \multicolumn{2}{c}{\textbf{12.06\%}}
& \multicolumn{2}{c}{\textbf{9.52\%}}\\
\midrule

\multirow{10}{*}{TKY}
& HR@1 $\uparrow$ &0.1170  &0.1131    &0.2316  &\textbf{0.2325}  &0.1457  &\textbf{0.1549}  &0.2153  &\textbf{0.2159}  &0.1541  &\textbf{0.1946} \\
& HR@10 $\uparrow$ &0.4482  &0.4437    &0.5599  &\textbf{0.5622}  &0.3950  &\textbf{0.4274}  &0.5068  &\textbf{0.5150}  &0.4807  &\textbf{0.4959} \\
& NDCG@10 $\uparrow$ &0.2656  &0.2628    &0.3866  &\textbf{0.3881}  &0.2627  &\textbf{0.2803}  &0.3551  &\textbf{0.3591}  &0.3050  &\textbf{0.3317} \\
\cmidrule(lr){2-12}
& Accuracy Improvement
& \multicolumn{2}{c}{-1.79\%}
& \multicolumn{2}{c}{\textbf{0.4\%}}
& \multicolumn{2}{c}{\textbf{7.07\%}}
& \multicolumn{2}{c}{\textbf{1.01\%}}
& \multicolumn{2}{c}{\textbf{12.73\%}} \\
\cmidrule(lr){2-12}
& CGF@1 $\downarrow$ &0.0591  &\textbf{0.0581}    &0.0603  &\textbf{0.0596}  &0.0602  &\textbf{0.0562}  &0.0597  &\textbf{0.0503}  &0.0392  &\textbf{0.0365} \\
& CGF@10 $\downarrow$ &0.0544  &\textbf{0.0536}    &0.0567  &\textbf{0.0561}  &0.0593  &\textbf{0.0579}  &0.0587  &\textbf{0.0523}  &0.0394  &\textbf{0.0380} \\
& CE@1 $\uparrow$ &0.0342  &\textbf{0.0366}    &0.0503  &\textbf{0.0531}  &0.0027  &\textbf{0.0079}  &0.0037  &\textbf{0.0041}  &0.0067  &\textbf{0.0080} \\
& CE@10 $\uparrow$ &0.2977  &\textbf{0.2984}    &0.0300  &\textbf{0.0313}  &0.0214  &\textbf{0.0271}  &0.0237  &\textbf{0.0252}  &0.0305  &\textbf{0.0345} \\
& LTC $\uparrow$ &0.3700  &\textbf{0.3798}    &0.7417  &\textbf{0.7451}  &0.2144  &\textbf{0.3432}   &0.2361  &\textbf{0.2441}  &0.6824  &\textbf{0.7176} \\
\cmidrule(lr){2-12}
& Fairness Improvement
& \multicolumn{2}{c}{\textbf{2.61\%}}
& \multicolumn{2}{c}{\textbf{2.52\%}}
& \multicolumn{2}{c}{\textbf{57.66\%}}
& \multicolumn{2}{c}{\textbf{9.44\%}}
& \multicolumn{2}{c}{\textbf{9.62\%}}\\
\midrule

\multirow{10}{*}{CA}
& HR@1 $\uparrow$ &0.0693  &\textbf{0.0752}  &0.1135    &\textbf{0.1158}  &0.1370  &\textbf{0.1890}  &0.0992  &\textbf{0.1162}  &0.1415  &\textbf{0.1677} \\
& HR@10 $\uparrow$ &0.2682  &0.2653  &0.2975   &\textbf{0.3024}  &0.3613  &\textbf{0.6059}  &0.3001  &\textbf{0.3084}  &0.3599  &\textbf{0.3662}  \\
& NDCG@10 $\uparrow$ &0.1606  &\textbf{0.1608}  &0.1975    &\textbf{0.1999}  &0.2379  &\textbf{0.3780}  &0.2063  &0.2030  &0.2416  &\textbf{0.2553}  \\
\cmidrule(lr){2-12}
& Accuracy Improvement
& \multicolumn{2}{c}{\textbf{2.52\%}}
& \multicolumn{2}{c}{\textbf{2.30\%}}
& \multicolumn{2}{c}{\textbf{54.85\%}}%757
& \multicolumn{2}{c}{\textbf{6.10\%}}
& \multicolumn{2}{c}{\textbf{8.64\%}} \\
\cmidrule(lr){2-12}
& CGF@1 $\downarrow$ &0.0269   &0.0275  &0.0421  &\textbf{0.0410}  &0.0473  &\textbf{0.0284}  &0.0452  &\textbf{0.0326}  &0.0218  &0.0253 \\
& CGF@10 $\downarrow$ &0.0255   &\textbf{0.0237}  &0.0383  &\textbf{ 0.0378}  &0.0497  &\textbf{0.0311}  &0.0458  &\textbf{0.0371}  &0.0194  &\textbf{0.0188}  \\
& CE@1 $\uparrow$ &0.0554  &\textbf{0.0673}    &0.0287  &\textbf{0.0307}  & 0.0072  &\textbf{0.0104}  &0.0024  &\textbf{0.0039}  &0.0042  &0.0036 \\
& CE@10 $\uparrow$ &0.2656  &\textbf{0.3481}    &0.1613  &\textbf{0.1654}  &0.0105  &\textbf{0.0419}  &0.0162  &\textbf{0.0209}  &0.0207  &\textbf{0.0330} \\
& LTC $\uparrow$ &0.3126  &\textbf{0.3443}  &0.6192    &\textbf{0.6444}  &0.3120  &\textbf{0.4672}  &0.1740  &\textbf{0.2103}  &0.5594  &\textbf{0.6081} \\
\cmidrule(lr){2-12}
& Fairness Improvement
& \multicolumn{2}{c}{\textbf{13.50\%}}
& \multicolumn{2}{c}{\textbf{3.50\%}}
& \multicolumn{2}{c}{\textbf{94.11\%}}
& \multicolumn{2}{c}{\textbf{31.85\%}}
& \multicolumn{2}{c}{\textbf{8.17\%}} \\
\bottomrule
\end{tabular}%
}
\end{table*}
\vspace{-2mm}
\subsection{Ablation Study}\label{Sec. 4.3}
To analyze the roles of different components, we conduct ablation studies on NYC and CA datasets. As similar conclusions are drawn across all backbone models, we present only the GETNext results in \tablename~\ref{tab:ablation on NYC} and \tablename~\ref{tab:ablation on CA}. Specifically, ``\textit{w/o} comm'' removes community inference, ``\textit{w/o} gen'' removes generation process, ``\textit{w/o} $\mathcal{L}_{\text{infer}}$'' denotes SPACE without community updating, and ``\textit{w/o} $\mathbf{s}^{(0)}$'' denots SPACE without physics anchor. We observe that removing any component leads to performance degradation, confirming the necessity of each module for both accuracy and fairness. On the denser NYC dataset, community prototypes and physical anchors serve as vital priors-their absence significantly impairs fairness. Conversely, virtual user generation is more effective for CA, mitigating its extreme sparsity and severe long-tail bias.
\begin{table}[t]
\centering
\caption{Ablation study on NYC.}
\vspace{-4pt}
\label{tab:ablation on NYC}
\setlength{\tabcolsep}{5pt}
\renewcommand{\arraystretch}{0.8}
\begin{tabular}{lccccc}
\toprule
\multirow{2}{*}{\textbf{Method}} &
\multicolumn{2}{c}{Accuracy $\uparrow$} &
\multicolumn{3}{c}{Fairness $\uparrow$} \\ 
\cmidrule(lr){2-3} \cmidrule(lr){4-6}
& \textbf{HR@1} & \textbf{HR@10} &\textbf{CE@1} & \textbf{CE@10} & \textbf{LTC}\\
\midrule
\textit{w/o} comm &0.2656  &0.6093    &0.0144  &0.0600  &0.5341 \\
\textit{w/o} gen     &0.2710  &0.6281   &0.0146  &0.0645  &0.5577  \\
\textit{w/o} $L_{\text{infer}}$       &0.2647  &0.6160    &0.0136  &0.0480  &0.4671 \\
\textit{w/o} $\mathbf{s}^{(0)}$ &0.2697  &0.6221     &0.0141  &0.0515  &0.4824 \\
\textbf{SPACE} & \textbf{0.2791} & \textbf{0.6638} & \textbf{0.0209} & \textbf{0.1100} & \textbf{0.7686}\\
\bottomrule
\end{tabular}
\end{table}
\begin{table}[t]
\centering
\caption{Ablation study on CA.}
\vspace{-4pt}
\label{tab:ablation on CA}
\setlength{\tabcolsep}{5pt}
\renewcommand{\arraystretch}{0.8}
\begin{tabular}{lccccc}
\toprule
\multirow{2}{*}{\textbf{Method}} &
\multicolumn{2}{c}{Accuracy $\uparrow$} &
\multicolumn{3}{c}{Fairness $\uparrow$} \\ 
\cmidrule(lr){2-3} \cmidrule(lr){4-6}
& \textbf{HR@1} & \textbf{HR@10} &\textbf{CE@1} & \textbf{CE@10} & \textbf{LTC}\\
\midrule
\textit{w/o} comm &\textbf{0.1906}  &0.5543    &0.0102  &0.0371  &0.4217 \\
\textit{w/o} gen     &0.1843  &0.5618   &0.0096  &0.0331  &0.4021  \\
\textit{w/o} $L_{\text{infer}}$       &0.1893  &0.5988   &\textbf{0.0106}  &0.0402  &0.4270 \\
\textit{w/o} $\mathbf{s}^{(0)}$  &0.1862  &0.5974    &0.0102  &0.0381  &0.4190 \\
\textbf{SPACE} & 0.1890 & \textbf{0.6059} & 0.0104 & \textbf{0.0419} & \textbf{0.4672}\\
\bottomrule
\end{tabular}
\end{table}
\vspace{-2mm}
\subsection{Hyperparameter Study}\label{Sec. 4.4}
We investigate the effect of different values of $\lambda_{virtual}$, $\alpha$, and $\lambda_{rec}$ on NYC, integrated with DiffuRec and GETNext in \figurename~\ref{fig: hyperparameter study}. The parameter $\lambda_{virtual}$ controls the loss of the community inference. With the driving of the training compensation signal, communities are divided effectively, and the performance is improved. But when $\lambda_{virtual}$ is too large, community inference is easily biased by the auxiliary signal, leading to a decline. The parameter $\alpha$ controls the disturbance of selecting physical anchors, which determines the balance between ``diversity'' and ``authenticity'' generated by virtual users. 
The trend of initial improvement followed by a decline as $\alpha$ increases illustrates that the generation mechanism needs moderate disturbance, rather than random disturbance. The parameter $\lambda_{rec}$ controls the weight of the virtual loss $L_{\text{virtual}}$ in total loss. A small $\lambda_{rec}$ underutilizes the fairness-aware signal, whereas a large $\lambda_{rec}$ may dominate training and hurt recommendation accuracy. Overall, SPACE remains stable across a broad range of hyperparameter settings, and its best performance is achieved under moderate parameter values.
\begin{figure}[t]
    \centering
    \vspace{-3mm} 
    \subfigure[Diffurec w/ SPACE]{
    \begin{minipage}[b]{1.0\linewidth}
    \centering
    \includegraphics[scale=0.53]{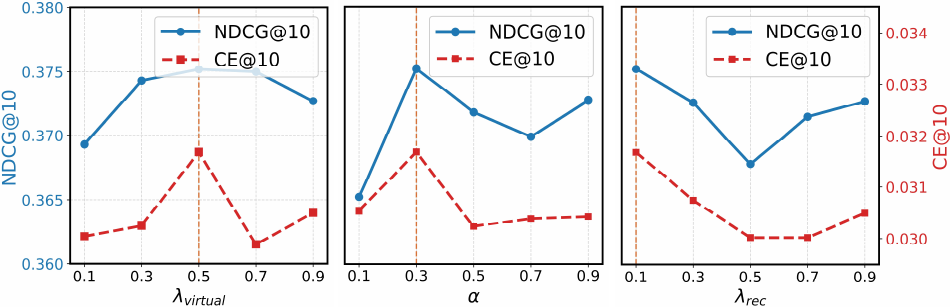}
    \end{minipage}
    }
    \vspace{-3mm}
    \subfigure[GETNext w/ SPACE]{
    \begin{minipage}[b]{1.0\linewidth}
    \centering  
    \includegraphics[scale=0.53]{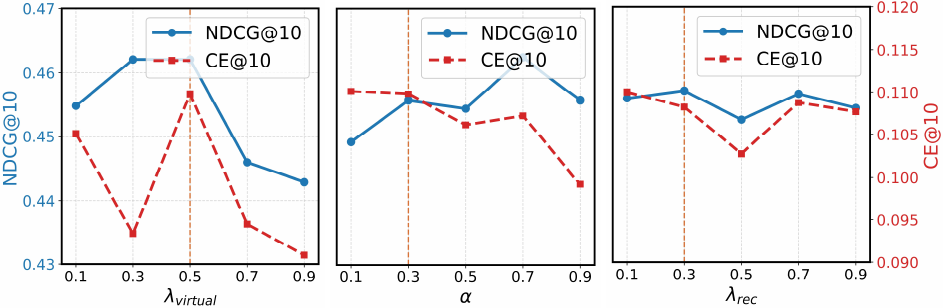}
    \end{minipage}
    }
\caption{Impact of hyperparameters on SPACE.}
\label{fig: hyperparameter study}
\end{figure}
\vspace{-2mm}
\subsection{Efficiency Analysis}\label{Sec. 4.5} 
\figurename~\ref{fig:Training Time} shows the computational overhead introduced by SPACE. Notably, SPACE introduces negligible additional overhead. This efficiency stems from the fact that once the community partitions are established, the generation of virtual users is a straightforward process. Given its ability to enhance performance and fairness, SPACE is well-suited for real-world applications.
\subsection{Supplement}\label{Sec. 4.6}
We also conduct experiments on two other metrics, NC@$K$ and GINI, as shown in \tablename~\ref{tab:supplement}. With the help of our SPACE, backbone models also achieve superior fairness performance.
%\vspace{-4mm}
\subsection{Case Study}\label{Sec. 4.7} 
In this section, we illustrate how our method affects recommendation behaviors for constrained users in \figurename~\ref{fig:case study}. Specifically, we consider three representative user types, including low-mobility commuters, nighttime users relying on public transportation, and exploration-oriented tourists, and report the Top-10 recommendation results of GETNext with and without SPACE. Obviously, with the help of SPACE, the recommendation results are generally closer to the users' real activity range and have better coverage of long-tail POIs, without deviating from the center of interest.
\begin{table}[t]
\centering
\caption{NC@$K$ measures the proportion of long-tail POI and the GINI measures the imbalance of global exposure.}
\vspace{-4pt}
\label{tab:supplement}
\setlength{\tabcolsep}{9pt}
\renewcommand{\arraystretch}{0.8}
\begin{tabular}{lcccc}
\toprule
\multirow{2}{*}{\textbf{Metric}} &
\multicolumn{2}{c}{GETNext} &
\multicolumn{2}{c}{MTNet} \\ 
\cmidrule(lr){2-3} \cmidrule(lr){4-5}
& w/o & w & w/o & w \\
\midrule
NC@1 $\uparrow$ &0.0356  &\textbf{0.1013}  &0.0492   &\textbf{0.0536}  \\
NC@10 $\uparrow$  &0.0673  &\textbf{0.0681}  &0.0726   &\textbf{0.0775}  \\
GINI $\downarrow$  &0.9236  &\textbf{0.9083}  &0.9210   &\textbf{0.9183}  \\
\midrule
Improvement
& \multicolumn{2}{c}{\textbf{62.46\%}}
& \multicolumn{2}{c}{\textbf{15.99\%}} \\
\bottomrule
\end{tabular}
\end{table}
\section{Related Work}

\subsection{Next POI Recommendation}

Early studies regard the visit history of a user as a sequence. Traditional methods based on Markov chains (MC) and matrix factorization (MF) utilize check-in sequences \cite{he2016inferring} and user-POI interactions \cite{salakhutdinov2008bayesian} for recommendation. FPMC \cite{rendle2010factorizing} integrates MF with MC to capture user preferences. 
% Subsequent work introduces RNNs and LSTMs to enhance the representation of check-in features and improve the understanding of user preferences. 
SASRec \cite{kang2018self} combines attention mechanisms with RNNs. HST-LSTM \cite{kong2018hst} and LSTPM \cite{zhao2020go} introduce LSTM to learn the influence of spatio-temporal or long-term information. STAN \cite{luo2021stan} further models non-adjacent check-ins through spatio-temporal interactions. To enhance the understanding of the unique spatio-temporal dependencies in POI scenes, recent studies explore diffusion models \cite{li2023diffurec,qin2023diffusion,zuo2024diff}, graph-based modeling \cite{yang2022getnext,yan2023spatio,yang2024siamese}, and LLM-based models \cite{li2024large,zhong2025comapoi}. Diff-POI \cite{qin2023diffusion} samples from a distribution with diffusion algorithm to recommend POIs in novel areas. GETNext \cite{yang2022getnext} and POIGDE \cite{yang2024siamese} integrate global trajectory flow map and the continuous evolution of user interests from graph perspective. LLM4POI \cite{li2024large} and CoMaPOI \cite{zhong2025comapoi} propose to integrate LLMs in POI recommendation tasks to address the limitations in semantic information fusion. 
% Although these methods contribute to improving recommendation performance, they overlook a crucial issue in POI scenarios: provider fairness. 

\vspace{-4mm}
\subsection{Provider Fairness in Recommendation}
Researchers \cite{zheng2017multi,CAPRI-FAIR} have recognized that recommendation platforms involve multiple stakeholders, including users \cite{chen2025leave,rampisela2025stairway}, item providers \cite{Multi-Objective,guo2025enhancing}, and the platform itself. As recommender systems increasingly cater to user intentions, a thought-provoking phenomenon has emerged: exposure bias\cite{guo2024configurable}, leading to unfair treatment across providers \cite{singh2018fairness,burke2018balanced}. To ensure provider fairness, some studies focus on mitigating popularity biases \cite{liu2025agsrec}, attempting to alleviate these biases by promoting long-tail items \cite{schnabel2016recommendations,wei2023collaborative} or applying causal interventions \cite{wang2021deconfounded,zhang2021causal}, thereby indirectly improving provider fairness. Another approach explicitly models fairness by introducing fairness measurement \cite{guo2024configurable} into the recommendation process, such as regularization-based techniques \cite{rhee2022countering,zhu2021popularity} and post-processing frameworks \cite{zhu2021fairness}. Recently, researchers propose to use LLMs and Agents to achieve provider fairness. FairAgent \cite{guo2025enhancing} builds a fairness enhancement framework based on reinforcement learning, using fairness metrics as rewards. 

% However, due to the unique constraints of POI scenarios, it is not feasible to directly apply provider fairness algorithms to protect disadvantaged POIs.
\begin{figure}[t]
    \centering
    \scalebox{1}[0.9]{%
    \includegraphics[height=4.5cm]{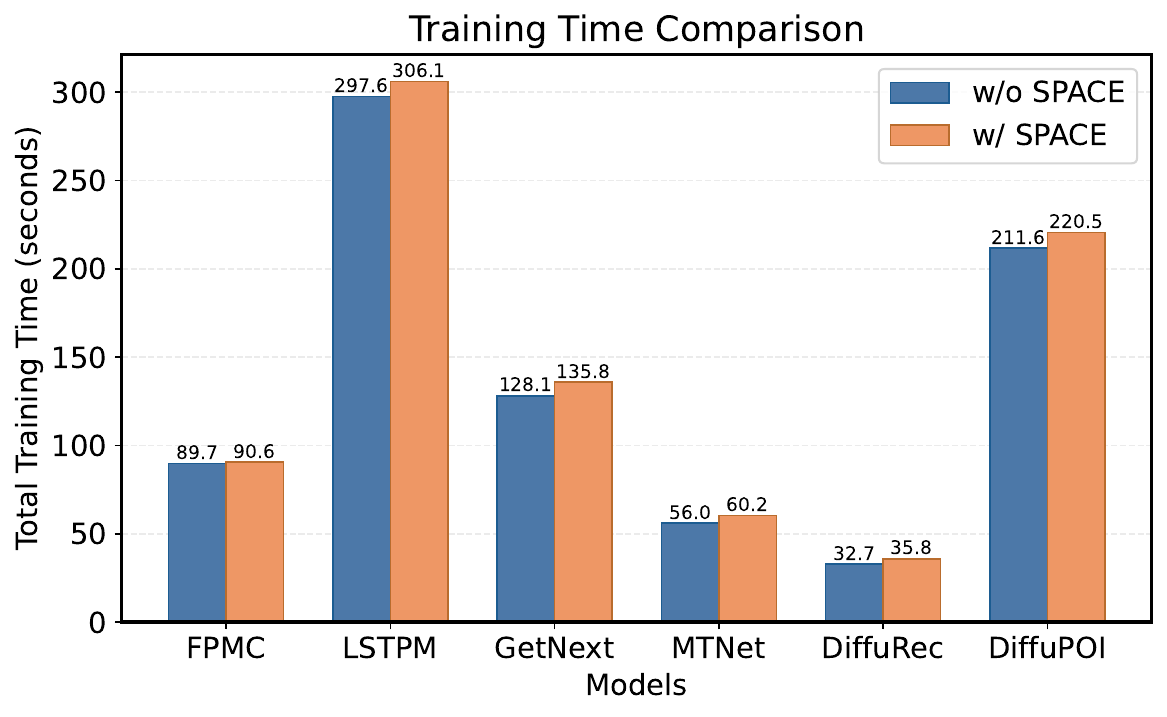}}
    \caption{Training time of models on the NYC dataset, for comparison with and without SPACE.}
    \label{fig:Training Time}
\end{figure}

\vspace{-3mm}
\section{Conclusion}
Unlike conventional settings, POI provider fairness is uniquely shaped by \emph{user execution constraints} and \emph{resource supply constraints}; ignoring them can yield superficial fairness gains at the expense of accuracy. 
In this work, we proposed \textsc{SPACE} (\textbf{S}upply- and \textbf{P}hysics-\textbf{A}ware \textbf{C}onditional \textbf{E}mbedding generation), a model-agnostic framework that democratizes long-tail POI exposure through constraint-respecting virtual user generation. \textsc{SPACE} first infers user communities to capture heterogeneous execution characteristics, then performs supply-constrained quota allocation for tail POIs via unbalanced optimal transport, and finally generates POI-conditional virtual user embeddings using constraint-guided latent diffusion. Extensive experiments on three real-world datasets validate the effectiveness of \textsc{SPACE}: it consistently enhances provider fairness for a range of backbone recommenders without sacrificing accuracy, and in most cases yields additional accuracy gains.

\begin{acks}
This work is supported by the National Natural Science Foundation of China under Grants No. 62572119 and 62232004, Jiangsu Provincial Key Laboratory of Network and Information Security under Grants No.BM2003201, Key Laboratory of Computer Network and Information Integration of Ministry of Education of China under Grants No.93K-9, the Fundamental Research Funds for the Central Universities No.2242026K30049, and partially supported by Collaborative Innovation Center of Novel Software Technology and Industrialization, Collaborative Innovation Center of Wireless Communications Technology. We also thank the Big Data Computing Center of Southeast University for providing the experiment environment and computing facility.
\end{acks}

\bibliographystyle{ACM-Reference-Format}
\bibliography{SPACE}

%\appendix

\end{document}